\documentclass{article}
\usepackage{amsmath, amssymb,amsthm,mathrsfs, graphicx}
\usepackage{url}

\begin{document}

\title{
A One-Factor Conditionally Linear 
Commodity Pricing Model under Partial Information\thanks{
This paper is based on the third author's master thesis \cite{YM}.}
}
\author{
Takashi Kato\thanks{
Division of Mathematical Science for Social Systems,
Graduate School of Engineering Science,
Osaka University. \ 
E-mail: \texttt{kato@sigmath.es.osaka-u.ac.jp}}
\and
Jun Sekine\thanks{
Division of Mathematical Science for Social Systems,
Graduate School of Engineering Science,
Osaka University. \ 
E-mail: \texttt{sekine@sigmath.es.osaka-u.ac.jp}}
\and
Hiromitsu Yamamoto
\thanks{
Division of Mathematical Science for Social Systems,
Graduate School of Engineering Science,
Osaka University. \ 
E-mail: \texttt{hiromitsu.yamamoto@gmail.com}}
}
\date{}
\maketitle

\begin{abstract}
A one-factor asset pricing model with
an Ornstein--Uhlenbeck process as its state variable 
is studied under partial information:
the mean-reverting level and the mean-reverting speed parameters 
are modeled as hidden/unobservable stochastic variables.
No-arbitrage pricing formulas
for derivative securities written on a liquid asset 
and  exponential utility indifference pricing formulas
for derivative securities written on an illiquid asset are presented.
Moreover, a conditionally linear filtering result 
is introduced to compute the pricing/hedging formulas
and the Bayesian estimators of the hidden variables. 

\medskip


\medskip


\noindent{\bf Journal of Economic Literature (JEL) Classifications: } 
G13, C63

\end{abstract}

\newtheorem{lem}{Lemma}[section]
\newtheorem{prop}{Proposition}[section]
\newtheorem{thm}{Theorem}[section]
\newtheorem{cor}{Corollary}[section]
\newtheorem{form}{Formula}[section]
\newtheorem{ass}{Assumption}[section]
\newtheorem{conj}{Conjecture}[section]
\theoremstyle{definition}
\newtheorem{defi}{Definition}[section]
\newtheorem{nota}{Notation}[section]
\newtheorem{cond}{Condition}[section]
\newtheorem{prob}{Problem}[section]
\newtheorem{ex}{Example}[section]
\newtheorem{exer}{Exercise}[section]
\newtheorem{que}{Question}[section]
\newtheorem{rem}{Remark}[section]
\numberwithin{equation}{section}

\section{Introduction}

Employing an Ornstein--Uhlenbeck process as the state variable 
is a simple and tractable way to model a commodity price process.
For example, in Schwartz (1997) a one-factor model for 
a commodity spot price process $(S_t)_{t\ge 0}$ is considered,
\begin{equation}
S_t = {\mathrm e}^{Y_t},
\quad
dY_t = -k(Y_t -l)dt + \sigma dW_t,
\quad
Y_0=\log S_0\in {\mathbb R},
\end{equation}
where 
$k$,$l$, and $\sigma\in {\mathbb R}_{++}(:=(0,\infty))$ are constant parameters,
and $W:=(W_t)_{t\ge 0}$ is a $1$-dimensional Brownian motion
on a filtered probability space
$(\Omega,{\mathcal F},{\mathbb P}, ({\mathcal F}_t)_{t\ge 0})$.
The probability ${\mathbb P}$ may be regarded as
the physical (i.e., real-world) probability 
or as the pricing (i.e., risk-neutral) probability.  
Also, in Schwartz (1998) 
a one-factor model with time-dependent volatility, 
\begin{equation}
S_t = {\mathrm e}^{Y_t},
\quad
dY_t = 
\left\{  r-c -\frac{\sigma(t)^2}{2}\right\}dt + \sigma(t) dW_t,
\quad
Y_0=\log S_0\in {\mathbb R}
\end{equation}
is studied under risk-neutral probability, 
where 
$r\in {\mathbb R}_{+}(:=[0,\infty))$, 
\[
\begin{split}
 c:=&\alpha-\frac{\sigma_2^2}{2\kappa^2}
+\frac{\rho\sigma_1\sigma_2}{\kappa}, \\
 \sigma(t)
:=&\sqrt{
\sigma_1^2
+\left( \sigma_2^2 -2\rho\sigma_1\sigma_2\right)
\frac{1-{\mathrm e}^{\kappa t}}{\kappa}
}
\end{split}
\]
with $\alpha, \kappa, \sigma_1,\sigma_2\in {\mathbb R}_{++}$ 
and $\rho\in [-1,1]$.
Here, the constant $r$ is interpreted as the risk-free interest rate,
and the constant $c$ is interpreted as the convenience yield. 
Hence, the futures (or forward) price process 
$(F_t)_{t\in [0,T_1]}$, delivering at $T_1\in {\mathbb R}_{++}$, is given by
\[
 F_t:=S_t {\mathrm e}^{(r-c)(T_1-t)},
\]
or, equivalently,
\[
 dF_t = F_t \sigma(t)dW_t,
\quad
F_0= S_0 {\mathrm e}^{(r-c)T_1}.
\]

In the present paper, 
inspired by Carmona and Ludkovski (2006), 
we aim to treat one-factor models such as (1.1) and (1.2)
under partial information setting.
As in Carmona and Ludkovski (2006), 
where a general commodity forward price model
is treated under partial information,
our model has the following features:
\begin{itemize}
 \item[(i)]
A futures (or a forward) is regarded as a liquid asset.

 \item[(ii)] The spot is regarded as an illiquid asset, 
and so the convenience yield is regarded as a hidden 
stochastic variable.

 \item[(iii)] Because of (ii), 
the pricing and hedging of derivatives written on the spot
is regarded as an incomplete market problem 
that contains the hidden variable.
\end{itemize}
In particular, we are interested in a simple, 
specific, ``conditionally linear'' example 
that was not studied in Carmona and Ludkovski (2006).
Under the physical probability, 
the state-variable (of a spot/futures price process) 
is given by
\begin{equation}
 dY_t= \left\{ f(t)+\Theta_0 -\Theta_1 Y_t \right\}dt
+ \sigma(t)dW_t.
\end{equation}
Here, $f$ and $\sigma$ are deterministic functions,
and both the parameter $\Theta_0$, 
which is interpreted as the convenience yield, 
and the mean-reversion speed parameter $\Theta_1$
are {\it unobservable} (hidden) random variables
that are estimated dynamically in a ``Bayesian'' way.
The model has the following 
interesting tractability and flexibility characteristics: 
\begin{itemize}
 \item[(a)] 
For pricing derivatives written on a liquid futures,
Black--Scholes pricing formula is applied
(see Proposition 3.1 and Corollary 3.1). 
 \item[(b)]
For pricing derivatives written on an illiquid spot, 
closed-form formulas of indifference prices
are provided (see Proposition 3.2, 
Remark 4.3, and Proposition 5.1, followed by Remark 5.2). 
 \item[(c)]
Under the physical probability measure, 
the log-price process of a futures or spot
is a mean-reverting ``OU-like'' process 
with a stochastic mean-reverting level/speed.
Explicit formulas of
the Bayesian estimators (i.e., filters) 
of these parameters and the convenience yield
are provided (see Proposition 4.1 and Remark 4.2).

 \item[(d)]
Under the physical probability measure, 
the model can be non-Gaussian in nature
(see Proposition 6.1 and 6.2).
\end{itemize}

The organization of the present paper is as follows.
The model is introduced in the next section, and
the pricing and hedging of derivatives is discussed in Section 3.
In Section 4, we introduce a filtering result, with which 
the dynamics of prices (of both futures and spots)
and the convenience yields are described using physical probability.
In Section 5, 
we compute the trivariate probability density function 
of the three-dimensional {\it Markovian} state
in a ``semi-explicit'' form under risk-neutral probability, 
which is useful for pricing/hedging computations.
In Section 6,
we compute the cumulants of the marginal distribution 
of the logarithmic futures price under physical probability, 
and in Section 7, we conclude.
All necessary proofs are collected in the appendix.

\section{Model}

Let $(\Omega,{\mathcal F},{\mathbb P})$
be a complete probability space 
endowed with a one-dimensional Brownian motion 
$W:=(W_t)_{t\ge 0}$ 
and the two-dimensional random variable 
$\Theta:=(\Theta_0,\Theta_1)^\top$, where $\Theta$ is independent of $W$
($(\cdot)^\top$ denotes the transpose of a vector or a matrix).
Let $({\mathcal F}_t)_{t\ge 0}$ be the filtration which is defined by
\[
 {\mathcal F}_t:=
\sigma\left( W_u; \ u\in [0,t]\right)
\vee \sigma(\Theta)\vee {\mathcal N},
\]
where ${\mathcal N}$ is the totality of the null sets.
For $T_1, F_0\in {\mathbb R}_{++}$, 
consider the solution 
$(Y_t)_{t\in [0,T_1]}$ to the 
following stochastic differential equation:
\begin{equation}
 dY_t = \left\{ f(t) + \Theta_0 -\Theta_1 Y_t\right\}dt
+ \sigma(t)dW_t,
\quad 
Y_0:=\log F_0
\end{equation}
on $(\Omega,{\mathcal F},{\mathbb P}, ({\mathcal F}_t)_{t\in [0,T_1]})$,
where $f,\sigma:[0,T_1] \to {\mathbb R}$ 
are continuous functions, so that $\sigma(\cdot)\ge \epsilon$
with $\epsilon>0$.
We define $(F_t)_{t\in [0,T_1]}$ by
\[
 F_t:={\mathrm e}^{Y_t},
\]
and call this the $T_1$-delivering futures (or forward)
price process of a commodity. 
By It\^o's formula, we see that
\[
 dF_t =F_t 
\left\{ 
\mu(t,Y_t,\Theta) dt
+\sigma(t)dW_t
\right\}, 
\]
where we set
\[
\mu(t,y,\Theta):=f(t) + \frac{\sigma(t)^2}{2}
+ \Theta_0 -\Theta_1 y.
\]
Regarding $\Theta_0$ as the convenience yield, we define 
the spot price process $(S_t)_{t\in [0,T_1]}$ by
\begin{equation}
 S_t:= F_t {\mathrm e}^{-(r-\Theta_0)(T_1-t)}. 
\end{equation}
As mentioned in the introduction, we assume that
the futures $F$ is a liquid asset 
and that the spot $S$ is an illiquid asset.
An agent's information flow is given by
the filtration $({\mathcal F}^F_t)_{t\in [0,T_1]}$, where
\begin{align*}
 {\mathcal F}_t^F:=&
\sigma\left( F_u; \ u\in [0,t]\right)
\vee {\mathcal N} \\
=&\sigma\left( Y_u; \ u\in [0,t]\right)
\vee {\mathcal N},
\end{align*}
which is generated by the liquid futures price process.
So, the convenience yield $\Theta_0$
and the mean-reverting speed parameter $\Theta_1$ of $F$
are hidden random variables for the agent 
and are estimated via the information flow
$({\mathcal F}^F_t)_{t\in [0,T_1]}$.

We introduce a measure change, which we will use later.
Let 
\[
 \lambda(t,Y_t,\Theta)
:=\frac{\mu(t,Y_t,\Theta)}{\sigma(t)}
\]
be the market price of risk at time $t$. Using this, 
we define the probability measure $\tilde{\mathbb P}$ 
on $(\Omega,{\mathcal F}_{T_1})$ by
\[
\frac{d\tilde{\mathbb P}}{d{\mathbb P}}
\biggm|_{{\mathcal F}_t}
=Z_t(\Theta), 
\]
where
\[
Z_t(\Theta)
:=\exp\left\{
-\int_0^t \lambda(u,Y_u,\Theta)dW_u
-\frac{1}{2}
\int_0^t \lambda(u,Y_u,\Theta)^2 du
\right\}.
\]
This is actually well-defined 
because the martingale property of $(Z_t(\Theta))_{t\in [0,T]}$
follows from the linear-growth property of $\lambda(t,Y_t,\Theta)$ 
with respect to $Y_t$. See Lemma 4.1.1 of Bensoussan (1990), for example. 
By the Cameron-Martin-Maruyama-Girsanov theorem, 
the process
\[
 \tilde{W}_t:=W_t + \int_0^t \lambda(u,Y_u,\Theta)du,
\quad t\in [0,T_1],
\]
is a $(\tilde{\mathbb P}, {\mathcal F}_t)$-Brownian motion, 
and the $\tilde{\mathbb P}$-dynamics of $F$ is expressed as
\begin{equation}
dF_t =F_t \sigma(t)d\tilde{W}_t,
\quad F_0\in {\mathbb R}_{++}.
\end{equation}
Moreover, we see the following.
\begin{lem}
{\rm (1)}  It holds that
\[
 {\mathcal F}^F_t = 
\sigma\left( \tilde{W}_u; \ u\in [0,t]\right)
\vee {\mathcal N}, 
\quad
t\in [0,T_1].
\]
So, the process $\tilde{W}$ is also a 
$(\tilde{\mathbb P}, {\mathcal F}_t^F)$-Brownian motion.

\noindent{\rm (2)}
$\tilde{W}$ and $\Theta$ are independent under $\tilde{\mathbb P}$.
The law of $\Theta$ under $\tilde{\mathbb P}$
is equal to that under ${\mathbb P}$.
\end{lem}
\begin{proof}
See Appendix.
\end{proof}

\section{Pricing and Hedging of Derivatives}

Let $T\in {\mathbb R}_{++}$ ($T\le T_1$) be a time horizon.
Consider an agent 
who dynamically trades the futures $F$ in a self-financing way. 
The cumulative gain process 
$(G_t(\pi))_{t\in [0,T]}$ of the agent 
is given by
\begin{equation}
 dG_t(\pi)
=\pi_t \frac{dF_t}{F_t}
+r G_t(\pi) dt,
\quad
G_0(\pi)=0.
\end{equation}
Here, 
$r\in {\mathbb R}_{+}$ is the risk-free interest rate and
$\pi:=(\pi_t)_{t\in [0,T]}$ is a dynamic trading strategy 
where $\pi_t$ represents the amount of money
invested in the futures at time $t$
(see Section 4.5 of Duffie and Richardson, 1991, for example). 
Let
\[
 {\mathscr A}_T:=\left\{
(p_t)_{t\in [0,T]} ; \
\text{${\mathcal F}_t^F$-progressively measurable, and }
\tilde{\mathbb E} \int_0^T p_t^2 dt<\infty
\right\}
\]
be the space of admissible investment strategies.
Combining (3.1) with (2.3), 
we see that, for $\pi\in {\mathscr A}_T$, 
\[
 G_t(\pi)=
{\mathrm e}^{rt}
\int_0^t {\mathrm e}^{-rs} \pi_s \sigma(s)d\tilde{W}_s,
\quad
t\in [0,T],
\]
and that
\[
\tilde{G}_t(\pi):={\mathrm e}^{-rt} G_t(\pi)
=\int_0^t {\mathrm e}^{-rs} \pi_s \sigma(s)d\tilde{W}_s,
\quad
t\in [0,T]
\]
is a $(\tilde{\mathbb P}, {\mathcal F}^F_t)$-martingale.

\subsection{Derivatives on Futures}

Consider the derivative security 
whose payoff at the maturity date $T\in {\mathbb R}_{++}$ 
is given by
\begin{equation}
 H\in L^2(\tilde{\mathbb P}, {\mathcal F}^F_T).
\end{equation}
Noting Lemma 2.1, 
we apply the Brownian martingale representation theorem
to see that 
there exists $\pi^H\in {\mathscr A}_T$ such that
\begin{equation}
\tilde{\mathbb E}[{\mathrm e}^{-rT}H|{\mathcal F}^F_t]
=\tilde{\mathbb E}[{\mathrm e}^{-rT}H]
+ \tilde{G}_t\left( \pi^H\right)
\quad 
t\in [0,T],
\end{equation}
where $\tilde{\mathbb E}[\cdot]$ and 
$\tilde{\mathbb E}[ \cdot | \cdot]$ denote
expectation and conditional expectation, respectively, for $\tilde{\mathbb P}$; thus, we can find the pair 
$(x^H,\pi^H)\in {\mathbb R}\times {\mathscr A}_T$ such that
\begin{equation}
 H= {\mathrm e}^{rT} x^H + G_T\left( \pi^H\right) .
\end{equation}
Here, the first term of the right-hand side of (3.4) is 
the $T$-value of the initial replication cost $x^H$
with continuously compounded interest rate $r$, 
and the second term of the right-hand side of (3.4) is
the cumulative gain of the futures trading up-to time $T$. 
We then apply a standard argument of no-arbitrage pricing
theory regarding complete markets to obtain the following.
\begin{prop} 
For the derivative security (3.2) maturing at $T$, 
the following assertions are valid. 
\begin{itemize}
 \item[\rm (1)] 
$\tilde{\mathbb E}[{\mathrm e}^{-r(T-t)}H|{\mathcal F}^F_t]$
is the no-arbitrage price of the derivative security at time 
$t\in [0,T]$. 

 \item[\rm (2)]
The initial cost $x^H:=\tilde{\mathbb E}[{\mathrm e}^{-rT}H]$
and the trading strategy $\pi^H\in {\mathscr A}_T$ 
that satisfies (3.4) is the hedging strategy with 
the minimal hedging cost. 
\end{itemize}
\end{prop}
When we consider a derivative security with the payoff
\begin{equation}
 H:=h(F_T)\in L^2(\tilde{\mathbb P}, {\mathcal F}_T^F)
\end{equation}
at maturity $T$, we obtain the following Black--Scholes pricing 
formula.
\begin{cor}
For the derivative security (3.5) maturing at $T$, 
the following assertions are valid. 
\begin{itemize}
 \item[\rm (1)]
$\tilde{\mathbb E}[{\mathrm e}^{-r(T-t)}H|{\mathcal F}^F_t]
=V^H(t,F_t)$, where
\begin{align*}
 V^H(t,x):=& {\mathrm e}^{-r(T-t)}
\int_{-\infty}^\infty
h\left( x {\mathrm e}^{\Sigma(t,T)z -\frac{1}{2}\Sigma(t,T)^2}\right)
\frac{1}{\sqrt{2\pi}}{\mathrm e}^{-\frac{z^2}{2}}dz, \\
\Sigma(t,T):=& \sqrt{\int_t^T \sigma(s)^2 ds}.
\end{align*}
 \item[\rm (2)] 
The relation (3.3) holds with 
\begin{align*}
\tilde{\mathbb E}[{\mathrm e}^{-rT}H]=V^H(0,F_0), \quad
\pi^H_t=\partial_x V^H(t,F_t)F_t,
\quad t\in [0,T).
\end{align*}
\end{itemize}
\end{cor}

\subsection{Indifference Pricing and Optimal Hedging}

We next consider a derivative security with the payoff 
\begin{equation}
\begin{split}
 &H=h(Y_T, \Theta), 
\quad\text{where} \\
&\text{$h(y,\theta):{\mathbb R}\times {\mathbb R}^2 \to {\mathbb R}$
is bounded and Borel-measurable}
\end{split}
\end{equation}
at maturity date $T\in {\mathbb R}_{++}$.
A typical example is the following.
\begin{ex}[European derivative on spot]
Consider 
\[
H= \tilde{h}(S_T),
\]
where $\tilde{h}:{\mathbb R}_{++}\to {\mathbb R}$
is bounded and Borel-measurable. 
Here, $S_T$ is the spot price at time $T$. By (2.2), we can write 
that
\[
 H= \tilde{h}\left( F_T {\mathrm e}^{-(r-\Theta_0)T}\right)
=h(Y_T, \Theta),
\]
where 
\[
 h(y,\theta):= \tilde{h}\left({\mathrm e}^{y -(r-\theta_0)T}\right).
\]
\end{ex}

Since the derivative security (3.6) is not replicable in general, 
that is, 
there does not exist a pair 
$(x^H,\pi^H)\in {\mathbb R}\times {\mathscr A}_T$, 
which satisfies the relation (3.4) in general, 
we cannot apply the no-arbitrage pricing theory in a complete market. 
Instead, we employ the exponential utility indifference price 
for the derivatives. 
Let 
\[
 U(x):= - {\mathrm e}^{-\gamma x}
\]
be the exponential utility function with risk-aversion parameter 
$\gamma>0$ and consider 
\[
V(x,H):= \sup_{\pi\in {\mathscr A}_T}
{\mathbb E} U\left(-H+{\mathrm e}^{rT}x + G_T(\pi) \right).
\]
Recall that we can write that
\begin{align}
 V(x,H)=& 
-\inf_{\pi\in {\mathscr A}_T}
{\mathbb E} \exp\left\{ 
-\gamma \left(-H+{\mathrm e}^{rT}x + G_T(\pi) \right)
\right\} \nonumber \\
=&-\exp\left( -\gamma {\mathrm e}^{rT}x\right)
\inf_{\pi\in {\mathscr A}_T}
{\mathbb E} \exp\left\{\gamma \left(H-G_T(\pi) \right) \right\}.
\end{align}
The indifference price $\hat{p}^H$ 
of the derivative $H$ at time $0$ is defined by the relation
\begin{equation}
V(x+\hat{p}^H,H)=V(x,0).
\end{equation}
Combining (3.8) with (3.7), we see that
\begin{equation}
\hat{p}^H
= \frac{{\mathrm e}^{-rT}}{\gamma}
\left\{
\inf_{\pi\in {\mathscr A}_T} \log 
{\mathbb E} \left[ 
{\mathrm e}^{\gamma \left(H-G_T(\pi)\right)}
\right]
-\inf_{\pi\in {\mathscr A}_T} \log 
{\mathbb E} \left[ 
{\mathrm e}^{-\gamma G_T(\pi)}
\right]
\right\}.
\end{equation}
We call the strategy $\hat{\pi}^H \in {\mathscr A}_T$ that satisfies
\begin{equation}
 V(x,H)= {\mathbb E}
U\left( -H+ {\mathrm e}^{rT}x + G_T(\hat{\pi}^H)\right)
\end{equation}
the optimal hedging strategy. 
We obtain the following.
\begin{prop}
For the derivative security (3.6), let
\begin{align}
 \hat{H}_T^{(\gamma)} :=& \frac{1}{\gamma} \log 
{\mathbb E}\left[ {\mathrm e}^{\gamma H} 
\bigm| {\mathcal F}^F_T\right], \\
 \tilde{H}_T^{(\gamma)}
:=& \frac{1}{\gamma}\log 
\tilde{\mathbb E}\left[
Z_T(\Theta)^{-1}{\mathrm e}^{\gamma H}
\bigm| {\mathcal F}^F_T \right].
\end{align}
Then, the following assertions are valid.
\begin{itemize}
 \item[\rm (1)]
The utility indifference price is equal to 
the no-arbitrage price of $\hat{H}_T^{(\gamma)}$:
\[
 \hat{p}^H= \tilde{\mathbb E}
\left[ {\mathrm e}^{-rT}\hat{H}_T^{(\gamma)} \right].
\]

 \item[\rm (2)]
The replicating strategy of $\tilde{H}_T^{(\gamma)}$, 
that is, the $\hat{\pi}^H\in {\mathscr A}_T$ that satisfies
\begin{equation}
\tilde{\mathbb E}\left[ {\mathrm e}^{-rT}\tilde{H}_T^{(\gamma)}
\bigm| {\mathcal F}^F_t\right]  
=
\tilde{\mathbb E}\left[ {\mathrm e}^{-rT} \tilde{H}_T^{(\gamma)} \right]
+ \tilde{G}_t\left(\hat{\pi}^H\right),
\quad
t\in [0,T],
\end{equation}
is the optimal hedging strategy, which satisfies (3.10).
\end{itemize}
\end{prop}
\begin{proof}
We can compare this with the results in Section 4 of Carmona and Ludkovski (2006) and also with Becherer (2003) and Henderson (2002) cited therein. For completeness, we give a direct proof in the appendix. 
\end{proof}

\begin{rem}
A related study, Mellios and Six (2011), 
should be mentioned. Mellios and Six (2011) 
employ an unobservable stochastic convenience yield model 
and studies an optimal hedging problem 
for a fixed commodity spot position  
using a commodity futures 
and a zero-coupon bond as the hedging instruments.
In contrast with our setting, 
the information flow for the hedger 
is generated by the commodity spot price process 
in the study.
\end{rem}

\section{A Conditionally Linear Filtering}

Hereafter, we assume that
\begin{equation}
\text{$f$ and $\sigma$ are constants}
\end{equation}
for simplicity.
In this section, we introduce filtering results 
for the hidden/unobservable random variable $\Theta$. 
For $g: {\mathbb R}^2 \to {\mathbb R}$, 
which is bounded and Borel measurable, let
\[
\widehat{g(\Theta)}_t
:={\mathbb E}\left[g(\Theta)| {\mathcal F}^F_t \right].
\]
By Bayes' rule, we see that
\begin{equation}
\widehat{g(\Theta)}_t
=\frac{\tilde{\mathbb E} \left[ Z_t(\Theta)^{-1} 
g(\Theta)\bigm| {\mathcal F}^F_t\right]}
{\tilde{\mathbb E}\left[ Z_t(\Theta)^{-1} \bigm| {\mathcal F}^F_t\right]}.
\end{equation}
Computing the right-hand side of (4.2), 
we obtain the following.
\begin{prop}
It holds that for $t\ge 0$,
\[
Z_t^{-1}(\Theta)
=\Lambda(\Theta; t,Y_t,P_t,Q_t),
\]
where we define 
\[
P_t:= \int_0^t Y_u du, \quad
Q_t:= \int_0^t Y_u^2 du,
\]
and
\begin{multline*}
\Lambda(\theta; t,y,p, q):=
\exp\left[
\frac{1}{\sigma^2}
(\theta_0+\alpha, -\theta_1) 
\begin{pmatrix}
y-\log F_0 +\frac{\sigma^2}{2}t \\
\frac{1}{2}\left( y^2-\log F_0^2 -\sigma^2 t +\sigma^2 p \right)
\end{pmatrix}
\right. \\
\left. 
-\frac{1}{2\sigma^2}
(\theta_0+\alpha, -\theta_1)
\begin{pmatrix}
t& p \\
p& q
\end{pmatrix}
\begin{pmatrix}
\theta_0+\alpha \\ -\theta_1
\end{pmatrix}
\right]
\end{multline*}
with $\alpha:=f+\sigma^2/2$.
So, from (4.2), it follows that
\[
\widehat{g(\Theta)}_t
=\int_{{\mathbb R}^2}g(\theta) \rho_t(d\theta),
\]
where we write the posterior probability as
\[
 \rho_t(d\theta):=
\frac{\Lambda(\theta; t,Y_t,P_t,Q_t)\nu(d\theta)}
{\int_{{\mathbb R}^2}\Lambda(\theta; t,Y_t,P_t,Q_t) \nu(d\theta)}
\]
and denote the prior distribution of $\Theta$ by $\nu$.
\end{prop}
\begin{proof}
See Appendix.
\end{proof}

\begin{rem}
Proposition 4.1 can be interpreted as 
a specific example of a conditionally almost linear filtering result; 
this is explored in Haussmann and Pardoux (1988)
in a general setting.
\end{rem}
\begin{rem}
For $t\ge 0$,
we can write the Bayes estimator 
${\mathbb E}[\Theta_i| {\mathcal F}^F_t]$
of $\Theta_i$ ($i\in \{0,1\}$) as
\begin{align*}
\hat{\Theta}_i(t):=
 {\mathbb E}[\Theta_i| {\mathcal F}^F_t]
=&
\frac{\int_{{\mathbb R}^2} \theta_i 
\Lambda(\theta; t,Y_t,P_t,Q_t) \nu(d\theta)}
{\int_{{\mathbb R}^2} \Lambda(\theta; t,Y_t,P_t,Q_t) \nu(d\theta)} \\
=&:\bar{\Theta}_i \left( t, Y_t, P_t,Q_t \right).
\end{align*}
Then, we note that the filtered convenience yield
$(\hat{\Theta}_0(t))_{t \in [0,T]}$
is a stochastic {\it process}, 
which is ${\mathcal F}^F_t$-adapted.
Here, we may assume that the support of $\nu$ is bounded, 
for example. 
We can describe the
$({\mathbb P}, {\mathcal F}^F_t)$-dynamics of 
the process $Y$ as
\begin{multline*}
dY_t =\left\{
f+\bar{\Theta}_0 \left( t, Y_t, \int_0^t Y_udu, \int_0^t Y_u^2 du \right) 
\right. \\
\left. 
-\bar{\Theta}_1 \left( t, Y_t, \int_0^t Y_udu, \int_0^t Y_u^2 du \right)
 Y_t
\right\} dt 
+\sigma dB_t
\end{multline*}
on $(\Omega,{\mathcal F}, {\mathbb P}, ({\mathcal F}^F_t)_{t\in [0,T_1]})$.
Here, $(B_t)_{t\in [0,T_1]}$ is the 
$({\mathbb P}, {\mathcal F}^F_t)$-Brownian motion 
(the so-called innovation process) defined by
\[
B_t:= 
\frac{1}{\sigma}
\left[
Y_t-Y_0
- \int_0^t 
\left\{
f+\hat{\Theta}_0(s)
-\hat{\Theta}_1(s)Y_s
\right\} ds
\right].
\]
\end{rem}

\begin{rem}
From Proposition 4.1, 
the random variables $\hat{H}^{(\gamma)}_T$ and 
$\tilde{H}^{(\gamma)}_T$, given by (3.11) and (3.12), respectively, 
can be represented as
\begin{align*}
 \hat{H}^{(\gamma)}_T
=& 
\frac{1}{\gamma} 
\left[
\log \int_{{\mathbb R}^2} 
{\mathrm e}^{\gamma h(Y_T, \theta)} \Lambda(\theta, T,Y_T,P_T,Q_T)
\nu(d\theta) \right. \\
&\quad\left. 
-\log \int_{{\mathbb R}^2} \Lambda(\theta, T,Y_T,P_T,Q_T)\nu(d\theta)
\right] \\
=&: 
\hat{{\mathcal H}}^{(\gamma)}_T \left( Y_T, P_T, Q_T \right), \\
 \tilde{H}^{(\gamma)}_T=&
\frac{1}{\gamma} 
\log \int_{{\mathbb R}^2} 
{\mathrm e}^{\gamma h(Y_T, \theta)} 
\Lambda(\theta, T,Y_T,P_T,Q_T)\nu(d\theta) \\
=&: 
\tilde{{\mathcal H}}^{(\gamma)}_T \left( Y_T, P_T, Q_T \right).
\end{align*}
With the three-dimensional 
$(\tilde{\mathbb P}, {\mathcal F}^F_t)$-Markovian process
\[
Y_t := \log F_0-\frac{\sigma^2}{2}t + \sigma \tilde{W}_t,
\quad
P_t :=\int_0^t Y_u du,
\quad
Q_t := \int_0^t Y_u^2 du, 
\]
the indifference price is written as
\[
 \hat{p}^H
=\tilde{\mathbb E}\left[{\mathrm e}^{-rT}
\hat{{\mathcal H}}_T^{(\gamma)}(Y_T,P_T,Q_T)\right],
\]
and the optimal hedging strategy $\hat{\pi}^H\in {\mathscr A}_T$
is computed as 
\begin{equation}
 \hat{\pi}^H_t=
\partial_y \tilde{\mathcal V}^H(t,Y_t,P_t,Q_t),
\quad
t\in [0,T),
\end{equation}
where we let
\[
 \tilde{\mathcal V}^H(t,y,p,q):=
\tilde{\mathbb E}\left[
{\mathrm e}^{-r(T-t)}
\tilde{{\mathcal H}}^{(\gamma)}_T \left( Y_T, P_T, Q_T \right)
\bigm| Y_t=y, P_t=p, Q_t=q
\right]
\]
and assume that 
$\tilde{\mathcal V}^H(\cdot,\cdot,\cdot,\cdot)$
is smooth enough to apply It\^o's formula.
\end{rem}

\begin{ex}[Stochastic convenience yield and constant 
mean-reverting speed]
Suppose that 
\[
 \nu(d\theta_0,d\theta_1)
=\nu_0(d\theta_0)\otimes \delta_{\bar{\theta}_1}(d\theta_1),
\] 
where 
$\nu_0$ is the law of $\Theta_0$
and $\delta_{\bar{\theta}_1}$ 
with $\bar{\theta}_1\in {\mathbb R}$ is a Dirac measure.
That is, 
the convenience yield is a hidden stochastic random variable
with prior distribution $\nu_0$
and constant mean-reverting speed $\Theta_1\equiv \bar{\theta}_1$.
Then, the expression of the posterior probability 
is simplified. 
Indeed, we see that
\[
\Lambda(\theta_0,\bar{\theta}_1; t,y,p, q)=
\bar{\Lambda}(\theta_0; t,y,p)
\exp\left\{
-\frac{\bar{\theta}_1}{2\sigma^2}
\left( y^2-\log F_0^2 -\sigma^2 t +\sigma^2 p \right)
-\frac{\bar{\theta}_1^2}{2\sigma^2} q
\right\},
\]
where we define
\[
\bar{\Lambda}(\theta_0; t,y,p) 
:= \exp\left\{
\frac{(\theta_0+\alpha)}{\sigma^2}
\left( y +\bar{\theta}_1 p-\log F_0+\frac{\sigma^2}{2}t \right)
-\frac{(\theta_0+\alpha)^2}{2\sigma^2} t
\right\}.
\]
So, applying Proposition 4.1, we see that
\[
\rho_t(d\theta)
=
\frac{\bar{\Lambda}(\theta_0;t,Y_t,P_t)\nu_0(d\theta_0)}
{\int_{{\mathbb R}}\bar{\Lambda}(\theta_0;t,Y_t,P_t)\nu_0(d\theta_0)}
\otimes \delta_{\bar{\theta}_1}(d\theta_1),
\]
where the terms containing $Q_t$ have been canceled.
Similarly, the indifference price of the derivative 
with payoff $H:=h(Y_T,\theta_0)$ at maturity $T$ is 
simplified to
\[
 \hat{p}^H=\tilde{\mathbb E} \left[
{\mathrm e}^{-rT}
\hat{\mathcal H}_T^{(\gamma)}(Y_T,P_T)\right],
\]
where
\begin{multline*}
\hat{{\mathcal H}}^{(\gamma)}_T \left( Y_T, P_T\right)
:=\frac{1}{\gamma} 
\left[
\log \int_{{\mathbb R}} 
{\mathrm e}^{\gamma h(Y_T, \theta_0)} \bar{\Lambda}(\theta_0; T,Y_T,P_T)
\nu_0(d\theta_0) \right. \\
\left. -\log \int_{{\mathbb R}} \bar{\Lambda}(\theta_0; T,Y_T,P_T)
\nu_0(d\theta_0)
\right].
\end{multline*}
\end{ex}

\section{Trivariate Density}

In this section, we are interested in computing 
the trivariate density 
$\phi_t: {\mathbb R}^2 \times {\mathbb R}_{+} \to {\mathbb R}_+$
with
\[
\phi_t(y,p,q) dy dp dq
:= \tilde{\mathbb P}
\left( Y_t \in dy, P_t\in dp, Q_t\in dq\right).
\]
This density is useful 
for computing the indifference price and the optimal hedging 
strategy studied in Section 3. Indeed, from Remark 4.3, 
we have an integral representation
of the indifference price
\[
\hat{p}^H
={\mathrm e}^{-rT}\int_{{\mathbb R}^2 \times {\mathbb R}_+} 
\hat{{\mathcal H}}_T^{(\gamma)}(y,p,q)
\phi_T(y,p,q)\,dy\,dp\,dq
\]
and 
\begin{multline*}
 \tilde{\mathcal  V}^H(t,y,p,q) \\
={\mathrm e}^{-r(T-t)}
\int_{{\mathbb R}^2 \times {\mathbb R}_+} 
\tilde{{\mathcal H}}^{(\gamma)}_T (y+y_1,p+p_1,q+q_1)
\phi_{T-t}(y_1,p_1,q_1)dy_1\,dp_1\,dq_1,
\end{multline*}
using which the optimal hedging strategy $\hat{\pi}^H\in {\mathscr A}_T$
is represented as (4.3).
We obtain the following.
\begin{prop}
{\rm (1)}
For $(t,y,p,q)\in {\mathbb R}_+ \times {\mathbb R}
\times {\mathbb R}\times {\mathbb R}_+$, 
it holds that
\begin{multline*}
\phi_t(y,p,q)
=\frac{1}{\sigma^4}
\exp\left\{-\frac{1}{2}(y-\log F_0)-\frac{\sigma^2}{8}t \right\} \\
\times \psi_t\left( 
\frac{y-\log F_0}{\sigma},
\frac{p-(\log F_0) t}{\sigma},
\frac{q-2\sigma (\log F_0) p- (\log F_0)^2 t}{\sigma^2}
\right),
\end{multline*}
where we define
\[
\psi_t(x,y,z)\,dx\,dy\,dz
:=
 {\mathbb P}\left( W_t \in dx, 
\int_0^t W_s ds \in dy, \int_0^t W_s^2 ds \in dz\right).
\]

\medskip

\noindent{\rm (2)}
We write
\[
\psi_t(x,y,z)=\psi^{(1)}_t(x,y) \psi^{(2)}_t(z|x,y),
\]
where
\begin{align*}
\psi^{(1)}_t(x,y)\,dx\,dy
:=&{\mathbb P}\left( W_t \in dx, \int_0^t W_s ds\in dy\right), \\
\psi^{(2)}_t(z|x,y)dz
:=&{\mathbb P}\left( 
\int_0^t W_s^2 ds \in dz
\Bigm| 
W_t =x, \int_0^t W_s ds=y\right).
\end{align*}
Then, the following assertions are valid.
\begin{itemize}
 \item[\rm (i)] It holds that
\begin{equation}
\psi^{(1)}_t(x,y)
=\frac{1}{2\pi \sqrt{\det(A_1(t))}} 
\exp\left\{
-\frac{1}{2}(x,y)
A_1(t)^{-1}
\begin{pmatrix}
x \\ y
\end{pmatrix}
\right\},
\end{equation}
where
\[
A_1(t):=
\begin{pmatrix}
t& \frac{t^2}{2} \\
\frac{t^2}{2}& \frac{t^3}{3}
\end{pmatrix}.
\]
 \item[\rm (ii)]
It holds that
\begin{align}
\Gamma_t (\alpha|x,y)
:=& \int_0^\infty 
{\mathrm e}^{-\frac{\alpha^2}{2} z}
\psi^{(2)}_t(z|x,y) dz 
\nonumber \\
=&\sqrt{\frac{\det( A_1(t) )}
{\det(A_2(t,\alpha))\cosh(\alpha t)}} 
\nonumber \\
&\times\exp\left[
-\frac{1}{2} 
(x,y)  
\left\{A_2(t,\alpha)^{-1} -A_1(t)^{-1}\right\}
\begin{pmatrix}
x \\ y
\end{pmatrix}
\right],
\end{align}
where we define
\begin{equation}
A_2(t,\alpha):=
\begin{pmatrix}
\frac{\tanh(\alpha t)}{\alpha}&
\frac{1-\mathrm{sech}(\alpha t)}{\alpha^2} \\
\frac{1-\mathrm{sech}(\alpha t)}{\alpha^2}&
\frac{\alpha t-\tanh(\alpha t)}{\alpha^3}
\end{pmatrix}
\quad \text{for $\alpha>0$}
\end{equation}
and set
\begin{equation}
A_2(t,0) :=\lim_{\alpha \downarrow 0} A_2(t,\alpha)=A_1(t). 
\end{equation}
\end{itemize}
\end{prop}

\begin{proof}
See Appendix.
\end{proof}

\begin{rem}
Since we have not been able to find 
formula (5.2) in the existing literature, 
we introduce a proof of it in the appendix for the sake of completeness.
Formula (5.2) may be considered as an extension of
\[
{\mathbb E}
\left[ 
\exp\left( -\frac{\alpha^2}{2} \int_0^t W_u^2 du\right)
\biggm|
W_t=x
\right]
=\sqrt{\frac{\alpha t}{\sinh(\alpha t)}}
\exp\left[-\frac{x^2}{2t}
\left\{ \alpha t \coth (\alpha t)-1\right\}
\right],
\]
which is seen in (2.5) of Mansuy and Yor (2008), 
and 1.9.7 in p.168 of Borodin and Salminen (2002), 
for example.
\end{rem}

\begin{rem}
The explicit representation (5.2) of the 
conditional moment generating function is 
useful for (approximately) computing 
the conditional density $\psi^{(2)}(z| x,y)$:
we can apply 
Gram-Charlier expansion, 
Edgeworth expansion, 
or the saddle-point approximation
(see Hall (1992) and Jensen (1995), for example), 
at least formally.  
Or, we may define, in an appropriate way,  
the conditional Laplace transform
\[
\tilde{\Gamma}_t(\beta|x,y)
=\Gamma_t\left( (2\beta)^{\frac{1}{2}} \bigm|x,y\right)
\]
for $\beta\in {\mathbb C}$ 
to (numerically) compute the inverse Laplace transform:
\[
\psi^{(2)}_t(\cdot|x,y)
={\mathcal L}^{-1}
\left[
\tilde{\Gamma}_t(\cdot|x,y)
\right].
\]
\end{rem}

\section{Cumulants}

In this section, 
we observe a non-Gaussian nature 
in the logarithmic futures price $Y_t$ ($t\in [0,T_1]$)
under the physical probability measure
by analyzing its cumulants.
For simplicity, we assume that
\begin{equation}
\begin{split}
&\text{$(\Theta_0,\Theta_1)$ are bounded random variables, and} \\
&\text{$\Theta_1>\epsilon>0$ almost surely with some $\epsilon>0$.}
\end{split}
\end{equation}
The logarithmic futures price process,  
given by (2.1), is rewritten as
\begin{equation}
 dY_t =-\Theta_1\left( Y_t 
-\frac{\Theta_0+f}{\Theta_1} \right)dt 
+\sigma dW_t,
\quad
Y_0=\log F_0
\end{equation}
on $(\Omega, {\mathcal F}, {\mathbb P}, ({\mathcal F}_t)_{t\ge 0})$.
The conditional cumulant generating function
\[
K_{s,t}(\alpha)
= \log {\mathbb E}
\left[ {\mathrm e}^{\alpha Y_t} \bigm| {\mathcal F}^F_s\right],
\]
where $0\le s\le t\le T_1$,  is computed as
\begin{align*}
K_{s,t}(\alpha)
=& \log {\mathbb E} \left[ {\mathbb E}
\left[ {\mathrm e}^{\alpha Y_t} \bigm| {\mathcal F}_s\right]
| {\mathcal F}^F_s\right] \\
=&
{\mathbb E} 
\left[
\exp\left\{
\alpha m_{t-s}(\Theta)
+\frac{\alpha^2}{2} v_{t-s}(\Theta)
\right\}
\biggm| {\mathcal F}^F_s\right],
\end{align*}
where we use the Gaussian property 
of $Y_t$ under the conditional probability 
${\mathbb P}( \, \cdot \, | {\mathcal F}_s)$, and
\begin{align}
m_{t-s}(\Theta):=& {\mathbb E}[Y_t|{\mathcal F}_s]
={\mathrm e}^{-\Theta_1 (t-s)} Y_s
+\frac{\Theta_0+f}{\Theta_1}\left( 1-{\mathrm e}^{-\Theta_1 (t-s)}\right), \\
v_{t-s}(\Theta):=& {\mathbb V}[Y_t|{\mathcal F}_s]
=\frac{\sigma^2}{2\Theta_1}
\left( 1-{\mathrm e}^{-2\Theta_1 (t-s)}\right)
\end{align}
with the notation ${\mathbb V}[\,\cdot \, | \, \cdot\, ]$
for the conditional variance (under ${\mathbb P}$).
So, from Proposition 4.1, we see that
\[
K_{s,t}(\alpha)
=\int_{{\mathbb R}^2}
\exp\left\{
\alpha m_{t-s}(\theta)
+\frac{\alpha^2}{2} v_{t-s}(\theta)
\right\}
\rho_s(d\theta).
\]
Then, by setting $s=0$, 
the unconditional cumulant generating function is 
written as
\begin{equation}
K_{0,t}(\alpha)
=\int_{{\mathbb R}^2}
\exp\left\{
\alpha m_{t}(\theta)
+\frac{\alpha^2}{2} v_{t}(\theta)
\right\}
\nu(d\theta),
\end{equation}
where $\nu\equiv \rho_0$ is the prior distribution of $\Theta$. 
For the unconditional cumulants
\[
\kappa_n (t):=
\partial_\alpha^n K_{0,t}(\alpha)|_{\alpha=0},
\quad n\in {\mathbb N}, 
\]
we see the following.
\begin{prop}
It holds that, with (6.3) and (6.4), 
\begin{align*}
\kappa_1(t)=&{\mathbb E}[m_t(\Theta)], \\
\kappa_2(t)=&{\mathbb E}[v_t(\Theta)]+{\mathbb V}[m_t(\Theta)], \\
\kappa_3(t)=&
{\mathbb E}\left[ \left( m_t(\Theta)-\kappa_1(t)\right)^3\right]
+3{\mathbb C}[m_t(\Theta),v_t(\Theta)], 
\quad \text{and}\\
\kappa_4(t)=&
{\mathbb E}\left[ \left( m_t(\Theta)-\kappa_1(t)\right)^4\right]
+3\left\{
{\mathbb V}[v_t(\Theta)]-{\mathbb V}[m_t(\Theta)]^2
\right\} \\
&+6{\mathbb C}\left[ m_t(\Theta)^2, v_t(\Theta)\right]
-12{\mathbb E}[m_t(\Theta)]
{\mathbb C}\left[ m_t(\Theta), v_t(\Theta)\right],
\end{align*}
where 
${\mathbb V}[\cdot]$
denotes variance and 
${\mathbb C}[\cdot, \cdot]$
denotes covariance. 
\end{prop}
\begin{proof}
Direct calculations from (6.5).
\end{proof}

When we consider the long-time limit
$\kappa_n(\infty):=\lim_{t\to\infty}\kappa_n(t)$
of cumulants, 
the dependence 
of the prior distribution of $\Theta$ 
on the cumulants 
becomes simpler and clearer, as follows.
\begin{prop}
In addition to (6.1), 
assume that
the mean-reversion speed $\Theta_1$ and 
the mean-reversion level
\[
 \Theta_2:=\frac{\Theta_0+f}{\Theta_1}
\]
of (6.2) are independent.
It then holds that
\[
K_{0,\infty}(\alpha)
:=\lim_{t\to\infty} K_{0,t}(\alpha)
= K^{(1)} \left( \frac{\alpha^2\sigma^2}{4}\right)
+K^{(2)}(\alpha), 
\]
where we define
\[
\begin{split}
K^{(1)}(\alpha)
:=&\log {\mathbb E}\exp\left( \alpha \Theta_1^{-1} \right), \\
K^{(2)}(\alpha)
:=&\log {\mathbb E}\exp\left( \alpha \Theta_2 \right).
\end{split} 
\]
Further, it follows that, for $n\in {\mathbb N}$, 
\begin{align*}
\kappa_{2n-1}(\infty)
=&\kappa^{(2)}_{2n-1}, \\
\kappa_{2n}(\infty)
=&
(2n-1)!!\left( \frac{\sigma^2}{2}\right)^n
\kappa^{(1)}_n
+\kappa^{(2)}_{2n}
\end{align*}
with
\[
\kappa^{(i)}_{n}:=\partial^n_\alpha K^{(i)} (\alpha)\bigm|_{\alpha=0}
\quad i=1,2.
\]
\end{prop}

\begin{proof}
See Appendix.
\end{proof}

\section{Conclusion}

In this paper, 
we have studied a one-factor commodity pricing model
under partial information.
Concretely, the state variable of the model is an Ornstein--Uhlenbeck process
in which the mean-reverting level and the mean-reverting speed parameters 
are modeled as hidden, unobservable random variables.
Using the model, we have provided 
no-arbitrage pricing formulas
for derivative securities written on a liquid 
commodity futures 
and the exponential utility indifference pricing formulas
for derivative securities written on an illiquid commodity spot.
Also, we have introduced a related conditionally linear filtering
result, which is useful 
for computing the pricing/hedging formulas
and Bayesian estimators of the hidden variables. 

Studying a multifactor generalization
would be an interesting and important future research topic
related with this work.
Indeed, multifactor modeling is more natural and suitable
for describing a rich term structure of futures/forwards.
We refer interested readers to the following as a starting point: 
Gibson and Schwartz (1990), Schwartz (1997), 
Yamauchi (2002),  
Akahori, Yasutomi and Yokota (2005), 
Casassus and Collin-Dufresne (2005), 
Carmona and Ludkowski (2006), Mellios and Six (2011), 
Shiraya and Takahashi (2012), 
and the references therein. 
A straightforward generalization of our model is the
following: 
let $Y:=(Y_t)_{t\ge 0}$ be the $n$-dimensional {\it observable} 
state variable governed by
\[
dY_t = 
\left\{ f(t) + \Theta_0 -\Theta_1 Y_t\right\} dt
+\sigma(t) dW_t,
\quad
Y_0 \in {\mathbb R}^n
\]
on $(\Omega, {\mathcal F}, {\mathbb P}, ({\mathcal F}_t)_{t\ge 0})$,
endowed with the 
$n$-dimensional Brownian motion $W:=(W_t)_{t\ge 0}$, 
and the $n$-dimensional $\Theta_0$ 
and $n\times n$-dimensional $\Theta_1$, 
both of which are random variables independent of $W$.
Here, ${\mathbb P}$ is regarded as the physical probability measure, 
\[
 {\mathcal F}_t:=\sigma( W_u; u\le t)
\vee \sigma(\Theta_0,\Theta_1),
\]
and both $f: {\mathbb R}_+ \to {\mathbb R}^n$
and $\sigma: {\mathbb R}_+ \to {\mathbb R}^{n\times n}$
are deterministic functions. 
Using this state variable, the $m$ futures price processes 
$F^i:=(F^i_t)_{t\in [0,T_i]}$
($i=1,\cdots,m$, $m\le n$, with $T_i$ the delivery date), 
are given by
\[
 F^i_t:={\mathrm e}^{Y^i_t},
\quad i=1,\cdots, m.
\]
We employ the filtration
\[
 {\mathcal F}^Y_t:=\sigma(Y_u; u\in [0,t])
\quad t\ge 0
\]
as the information flow of an agent 
and regard $(\Theta_0,\Theta_1)$ 
as a {\it hidden/unobservable} variable
that is estimated in by Bayesian methods.
We note that the models fitting this multifactor formulation 
with the conditionally linear Gaussian state variable $Y$
include the models employed in 
Gibson and Schwartz (1990), Schwartz (1997), 
Yamauchi (2002), Casassus and Collin-Dufresne (2005), and
Shiraya and Takahashi (2012)
by setting $(\Theta_0,\Theta_1)$ as constants 
and choosing linear Gaussian state variables.

\appendix

\section{Collected Proofs}

\subsection{Proof of Lemma 2.1}

{\rm (1)} 
The assertion follows from the relations
\begin{align*}
Y_t =&Y_0
-\int_0^t \frac{\sigma(u)^2}{2}du
+\int_0^t \sigma(u) d\tilde{W}_u,  \\
 \tilde{W}_t =& \int_0^t \frac{dY_u}{\sigma(u)}
-\int_0^t \frac{\sigma(u)}{2}du.
\end{align*}

\noindent{\rm (2)}
Since $\tilde{W}$ is a $(\tilde{\mathbb P}, {\mathcal F}_t)$-Brownian
motion, 
we see that, for any 
$0=t_0<t_1<\cdots<t_n=T_1$, 
the increments
$\tilde{W}_{t_{i}}-\tilde{W}_{t_{i-1}}$ 
($i=1,\ldots,n$) are independent of 
${\mathcal F}_0=\sigma(\Theta)\vee {\mathcal N}$.
Hence, $\tilde{W}$ and $\Theta$ are independent. 
Also, we see that
\[
 \tilde{\mathbb E}[h(\Theta)]
= {\mathbb E} [Z_0(\Theta)h(\Theta)]
={\mathbb E}[h(\Theta)]
\]
for any bounded measurable function $h$.

\subsection{Proof of Proposition 3.2}

For $\pi\in {\mathscr A}_T$, we see that
\begin{align}
\log {\mathbb E}
\left[ {\mathrm e}^{\gamma \{ H-G_T(\pi)\}} \right]
=&\log \tilde{\mathbb E} 
\left[ Z_T(\Theta)^{-1}{\mathrm e}^{\gamma \{ H-G_T(\pi)\}} \right]
\nonumber \\
=&\log \tilde{\mathbb E} 
\left[ 
\tilde{\mathbb E}\left[
Z_T(\Theta)^{-1}{\mathrm e}^{\gamma H}
\bigm| {\mathcal F}^F_T \right]
{\mathrm e}^{-\gamma G_T(\pi)} \right] 
\nonumber \\
=&\log \tilde{\mathbb E} \left[
{\mathrm e}^{\gamma \{ \tilde{H}^{(\gamma)}_T -G_T(\pi)\}}\right] 
\nonumber \\
\ge& \gamma \tilde{\mathbb E}
\left[ \tilde{H}^{(\gamma)}_T -G_T(\pi) \right]
= \gamma \tilde{\mathbb E}\left[ \tilde{H}^{(\gamma)}_T \right],
\end{align}
where we use Jensen's inequality to derive the inequality 
in (A.1).
By the Brownian martingale representation theorem, 
we see that there exists $\hat{\pi}^H\in {\mathscr A}_T$
that satisfies (3.13).
The strategy satisfies
\begin{equation}
\log \tilde{\mathbb E} \left[
{\mathrm e}^{\gamma \left( \tilde{H}^{(\gamma)}_T 
-G_T(\hat{\pi}^H)\right)}\right]
= \gamma \tilde{\mathbb E}\left[ \tilde{H}^{(\gamma)}_T \right].
\end{equation}
Combining (A.1) and (A.2), we have that
\begin{equation}
\inf_{\pi\in {\mathscr A}_T}
{\mathbb E} \left[ {\mathrm e}^{\gamma(H-G_T({\pi}))} \right]
={\mathbb E} \left[ {\mathrm e}^{\gamma(H-G_T(\hat{\pi}^H))} \right]
={\mathrm e}^{\gamma \tilde{\mathbb E}[\tilde{H}^{(\gamma)}_T]}.
\end{equation}
From (3.9) and (A.3), we deduce that
\begin{align*}
 \hat{p}^H
=&\frac{{\mathrm e}^{-rT}}{\gamma}
\tilde{\mathbb E}
\left[
\log 
\tilde{\mathbb E}
\left[ Z_T(\Theta)^{-1} {\mathrm e}^{\gamma H}
\bigm| {\mathcal F}^F_T\right]
-\log 
\tilde{\mathbb E}
\left[ Z_T(\Theta)^{-1}
\bigm| {\mathcal F}^F_T\right]
\right] \\
=& \tilde{\mathbb E}\left[ 
{\mathrm e}^{-rT} \hat{H}_T^{(\gamma)}
\right],
\end{align*}
where we use Bayes' rule.

\subsection{Proof of Proposition 4.1}

Recalling that
\[
dY_t = \sigma d\tilde{W}_t-\frac{\sigma^2}{2} dt,
\quad
d\tilde{W}_t =\frac{dY_t}{\sigma}+\frac{\sigma}{2}dt 
\]
and that
\[
Y_t dY_t = \frac{1}{2} 
\left( dY_t^2 -\sigma^2 dt \right),
\]
we see that, letting 
$\alpha:=f+\sigma^2/2$, 
\begin{align*}
&\log Z_t(\Theta)^{-1} \\
=& \frac{1}{\sigma}\int_0^t
\left( \alpha+\Theta_0 -\Theta_1 Y_u\right) d\tilde{W}_u
-\frac{1}{2\sigma^2}
\int_0^t \left| \alpha+\Theta_0 -\Theta_1 Y_u\right|^2 du \\
=& \frac{1}{\sigma^2}\int_0^t
\left( \alpha+\Theta_0 -\Theta_1 Y_u\right)
\left( dY_u +\frac{\sigma^2}{2}du\right) 
-\frac{1}{2\sigma^2}
 \int_0^t
\left| \alpha+\Theta_0 -\Theta_1 Y_u\right|^2 du \\
=& \log \Lambda(\Theta;t, Y_t, P_t, Q_t).
\end{align*}
From this, using Lemma 2.1, we obtain the expression for 
the posterior probability $\rho_t(d\theta)$. 

\subsection{Proof of Proposition 5.1}

\noindent{(1)} Recall that
\[
 Y_t =\log F_0 -\frac{\sigma^2}{2}t +\sigma \tilde{W}_t.
\]
So, writing $y_0=\log F_0$, we deduce that
\begin{align*}
&\phi_t(y,p,q) dydpdq \\
=&\tilde{\mathbb P}
\left( Y_t \in dy, 
\int_0^t Y_s ds \in dp, 
\int_0^t Y_s^2 ds \in dq\right) \\
=&{\mathbb E}\left[
\exp\left(-\frac{\sigma}{2}W_t-\frac{\sigma^2}{8}t \right)
1_{\left\{ 
y_0 +\sigma W_t \in dy, 
\int_0^t (y_0+\sigma W_s) ds \in dp, 
\int_0^t (y_0+\sigma W_s)^2 ds \in dq 
\right\}}
\right] \\
=&\exp\left\{-\frac{1}{2}(y-y_0)-\frac{\sigma^2}{8}t \right\}  \\
&\times {\mathbb P}
\left( 
y_0 +\sigma W_t \in dy, 
\int_0^t (y_0+\sigma W_s) ds \in dp, 
\int_0^t (y_0+\sigma W_s)^2 ds \in dq 
\right) \\
=&
\exp\left\{-\frac{1}{2}(y-y_0)-\frac{\sigma^2}{8}t \right\} \\
&\times \psi_t\left( 
\frac{y-y_0}{\sigma},
\frac{p-y_0 t}{\sigma},
\frac{q-2\sigma y_0 p- y_0^2 t}{\sigma^2}
\right)
\frac{dy}{\sigma}\frac{dp}{\sigma}\frac{dq}{\sigma^2},
\end{align*}
where we use the Cameron-Martin-Maruyama-Girsanov formula.

\medskip

\noindent{(2)} 
The truth of assertion (i) is easy to see. 
To show assertion (ii) is true, we use a lemma.
\begin{lem}
For $t, \alpha \ge 0$ and $\beta_1,\beta_2\in {\mathbb R}$,
it holds that
\begin{align*}
L_t(\alpha,\beta_1,\beta_2)
:=& {\mathbb E}
\exp\left( \beta_1 W_t + \beta_2 \int_0^t W_s ds 
-\frac{\alpha^2}{2} \int_0^t W_s^2 ds \right) \\
=&
\frac{1}{\sqrt{\cosh(\alpha t)}}
\exp\left\{
\frac{1}{2}
(\beta_1,\beta_2) A_2(t,\alpha)
\begin{pmatrix}
\beta_1 \\ \beta_2
\end{pmatrix}
\right\},
\end{align*}
where we use (5.3)-(5.4).
\end{lem}

\begin{proof}
If $\alpha=0$ then the assertion is seen to be true from (5.1) and (5.4).

Next, suppose $\alpha>0$. 
We see that
\begin{align*}
L_t(\alpha, \beta_1,\beta_2)
=&{\mathbb E}
\exp\left\{
-\int_0^t \left(\alpha W_s -\frac{\beta_2}{\alpha} \right) dW_s
-\frac{1}{2}\int_0^t \left(\alpha W_s -\frac{\beta_2}{\alpha} \right)^2 ds
\right\} \\
&\times\exp\left\{
\frac{\alpha}{2} (W_t^2-t)
+\left( \beta_1 -\frac{\beta_2}{\alpha}\right)W_t
+\frac{1}{2}\left( \frac{\beta_2}{\alpha}\right)^2 t 
\right\} \\
=&\exp\left[
\frac{1}{2}
\left\{
\left( \frac{\beta_2}{\alpha}\right)^2 
-\alpha
\right\}t
\right]
{\mathbb E}
\left[
\exp\left\{
\frac{\alpha}{2} X_t^2
+\left( \beta_1 -\frac{\beta_2}{\alpha}\right)X_t
\right\}
\right],
\end{align*}
where 
we use the Cameron-Martin-Maruyama-Girsanov formula
and set
\[
 dX_t= dW_t 
-\left( \alpha X_t -\frac{\beta_2}{\alpha}\right) dt,
\quad
X_0=0.
\]
Setting
\begin{align*}
m_t:=&{\mathbb E}X_t
= \frac{\beta_2}{\alpha^2}(1-{\mathrm e}^{-\alpha t}), \\
v_t:=&{\mathbb E}(X_t-m_t)^2
=\frac{1}{2\alpha}(1-{\mathrm e}^{-2\alpha t}),
\end{align*}
we have that
\begin{align*}
{\mathbb E}
\left[
\exp\left(
\frac{\alpha}{2} X_t^2 +\beta X_t
\right)
\right]  
=\frac{1}{\sqrt{1-\alpha v_t}}
\exp\left\{
\frac{\beta^2 v_t + 2\beta m_t +\alpha m_t^2}
{2(1-\alpha v_t)}
\right\}.
\end{align*}
Hence,
\begin{align*}
&\log {\mathbb E}
\left[
\exp\left\{
\frac{\alpha}{2} X_t^2
+\left( \beta_1 -\frac{\beta_2}{\alpha}\right)X_t
\right\}
\right]  \\
=&-\frac{1}{2} \log (1-\alpha v_t)
+\frac{1}{2(1-\alpha v_t)}
\left\{
\left( \beta_1 -\frac{\beta_2}{\alpha}\right)^2 v_t
+ 2\left( \beta_1 -\frac{\beta_2}{\alpha}\right)m_t
+ \alpha m_t^2
\right\} \\
=&-\frac{1}{2} \log (1-\alpha v_t) \\
&+\frac{1}{2(1-\alpha v_t)}
\left\{
\left( \beta_1 -\frac{\beta_2}{\alpha}\right)^2 v_t
+ \frac{2\beta_2}{\alpha^2}
\left( \beta_1 -\frac{\beta_2}{\alpha}\right)(1-{\mathrm e}^{-\alpha t})
+ \frac{\beta_2^2}{\alpha^3}
(1-{\mathrm e}^{-\alpha t})^2
\right\} \\
=&-\frac{1}{2} \log (1-\alpha v_t) 
+\frac{1}{2(1-\alpha v_t)}
\left[
v_t \beta_1^2
-\frac{v_t}{\alpha^2} \beta_2^2
+\frac{(1-{\mathrm e}^{-\alpha t})^2}{\alpha^2} \beta_1 \beta_2
\right].
\end{align*}
So, 
\begin{align*}
\log L_t(\alpha,\beta_1,\beta_2)
=&
-\frac{\alpha}{2}t
-\frac{1}{2} \log (1-\alpha v_t)  
+\frac{v_t}{2(1-\alpha v_t)} \beta_1^2 \\
&+\frac{1}{2\alpha^2}
\left( 
t-\frac{v_t}{1-\alpha v_t}
\right) \beta_2^2
+\frac{(1-{\mathrm e}^{-\alpha t})^2}{2\alpha^2(1-\alpha v_t)} 
\beta_1 \beta_2.
\end{align*}
Recalling that
\begin{align*}
1-\alpha v_t
=&{\mathrm e}^{-\alpha t} \cosh( \alpha t), \\
\frac{v_t}{1-\alpha v_t}
=&\frac{1}{\alpha} \tanh (\alpha t), \\
\frac{(1-{\mathrm e}^{-\alpha t})^2}{(1-\alpha v_t)} 
=&2 \left( 1-\mathrm{sech} (\alpha t)\right),
\end{align*}
we deduce that
\begin{align*}
\log L_t(\alpha,\beta_1,\beta_2)
=&-\frac{1}{2} \log \cosh(\alpha t)
+\frac{1}{2\alpha} \tanh (\alpha t) \beta_1^2 \\
&+\frac{1}{2\alpha^3}
\left\{ \alpha t - \tanh(\alpha t) \right\} \beta_2^2
+\frac{1}{\alpha^2} 
\left\{ 1-\mathrm{sech} (\alpha t)\right\}\beta_1\beta_2,
\end{align*}
and this completes the proof.
\end{proof}

We are now in a position to show that assertion (ii) is true.
We deduce that
\begin{align*}
&L_t(\alpha,\beta_1,\beta_2) \\
=& \frac{1}{2\pi}
\frac{1}{\sqrt{{\det(A_2(t,\alpha))\cosh(\alpha t)}}} \\
\times& \int_{{\mathbb R}^2} 
\exp\left\{
(\beta_1,\beta_2)
\begin{pmatrix}
x \\ y
\end{pmatrix}
-\frac{1}{2} 
(x,y) A_2(t,\alpha)^{-1} 
\begin{pmatrix}
x \\ y
\end{pmatrix}
\right\} dx\,dy \\
=&\sqrt{\frac{\det(A_1(t))}{\det(A_2(t,\alpha))\cosh(\alpha t)}} \\
\times &\int_{{\mathbb R}^2} 
\exp\left[
(\beta_1,\beta_2)
\begin{pmatrix}
x \\ y
\end{pmatrix}
-\frac{1}{2} 
(x,y) 
\left\{
A_2(t,\alpha)^{-1} 
-A_1(t)^{-1}\right\}
\begin{pmatrix}
x \\ y
\end{pmatrix}
\right] \psi^{(1)}_t(x,y)\,dx\,dy \\
=&
\int_{{\mathbb R}^2}
{\mathrm e}^{\beta_1 x +\beta_2 y}
\Gamma_t(\alpha|x,y) \psi^{(1)}_t(x,y)\,dx\,dy.
\end{align*}

\subsection{Proof of Proposition 6.2}

We see that
\begin{align*}
m_\infty(\Theta):=& \lim_{t\to\infty}m_t(\Theta)=\Theta_2, \\
v_\infty(\Theta):=& \lim_{t\to\infty}v_t(\Theta)=
\frac{\sigma^2}{2}\Theta_1^{-1}.
\end{align*}
So, using the dominated convergence theorem, we deduce that
\begin{align*}
 K_{0,\infty}(\alpha)
=&\log {\mathbb E}\exp\left\{ \alpha m_\infty(\Theta)
+\frac{\alpha^2}{2} v_\infty(\Theta)\right\} \nonumber \\
=& K^{(1)} \left( \frac{\alpha^2\sigma^2}{4}\right)
+K^{(2)}(\alpha).
\end{align*}
Further, we deduce that
\begin{align*}
\partial_\alpha^n K_{0,\infty}(\alpha) \bigm|_{\alpha=0}
=
\begin{cases}
\kappa^{(2)}_{2m-1}&
\text{if $n=2m-1$}, \\
(2m-1)!!\left( \frac{\sigma^2}{2}\right)^m
\kappa^{(1)}_m+\kappa^{(2)}_{2m}&
\text{if $n=2m$}, 
\end{cases}
\end{align*}
and that
$\kappa_n(t):=\partial^n_\alpha K_{0,t}(\alpha) \bigm|_{\alpha=0}$
is represented as a 
sum of products of the expected values of 
some polynomials of $m_t(\Theta)$ and $v_t(\Theta)$.
Therefore, 
we can apply the dominated convergence theorem to see that
\[
\kappa_n(\infty)
=\lim_{t\to\infty} 
\partial^{n}_\alpha K_{0,t}(\alpha)\bigm|_{\alpha=0}
=\partial^{n}_\alpha K_{0,\infty}(\alpha)\bigm|_{\alpha=0}.
\]

\section*{Acknowledgments}

Jun Sekine's research was supported by a Grant-in-Aid 
for Scientific Research (C), No. 23540133, 
from the Ministry of Education, Culture, 
Sports, Science, and Technology, Japan.
This article is the preprint version of the article 
``A One-Factor Conditionally Linear Commodity Pricing Model under Partial Information'' 
published in {\it Asia-Pacific Financial Markets}, 
DOI: 10.1007/s10690-014-9182-y. 
The final publication is available at link.springer.com: \\
\url{http://link.springer.com/article/10.1007/s10690-014-9182-y}


\begin{thebibliography}{99}
 \bibitem{AYY}
{\sc Akahori, J., K. Yasutomi, and T. Yokota} (2005): 
Backwardation in Asian option prices, 
{\it International Journal of Innovative Computing, 
Information and Control}, {\bf 1}(3), 581--593.

 \bibitem{Bec}
{\sc Becherer, D.} (2003): 
Rational hedging and valuation of integrated risks 
under constant absolute risk aversion,
{\it Insurance: Mathematics and Economics}, {\bf 33}(1), 1--28.

 \bibitem{Ben}
{\sc Bensoussan, A.} : 
{\it Stochastic Control of Partially Observable Systems}.
{Cambridge University Press}, 1992.

 \bibitem{BS}
{\sc Borodin, A. N.  and P. Salminen} :
{\it Handbook of Brownian Motion: Facts and Formulae}
(Second Edition). 
Birkhauser Verlag, 2002.

 \bibitem{CL1}
{\sc Carmona, R.  and M. Ludkovski} (2006): 
Pricing commodity derivatives with basis risk and 
partial observations. {\it Working paper}, downloadable from
{\tt http://www.pstat.ucsb.edu/faculty/ludkovski/papers.html}



\bibitem{CC}
{\sc Casassus, J. and P. Collin-Dufresne} (2005) :
Stochastic convenience yield implied
from commodity futures and interest rates.
{\it The Journal of Finance}, {\bf 60} (5), 2283--2331.

\bibitem{DR} {\sc Duffie, D. and H. R. Richardson} (1991): 
Mean-variance hedging in continuous time, 
{\it The Annals of Applied Probability}, 
{\bf 1} (1), 1--15.

\bibitem{GS}
{\sc Gibson, R. and E. S. Schwartz} (1990) :
Stochastic convenience yield and the pricing of oil contingent claims, 
{\it The Journal of Finance}, {\bf 45} (3), 959--976.

\bibitem{H} {\sc Hall, P.}: 
{\it The Bootstrap and Edgeworth Expansion}, 
{Springer Series in Statistics}, Springer-Verlag, 1992.

\bibitem{HP} {\sc Haussmann, U. G. and E. Pardoux} (1988): 
A conditionally almost linear filtering problem with
non-Gaussian initial condition.
\emph{Stochastics}, {\bf 23}, 241--275.

\bibitem{Hen}
{\sc Henderson, V.} (2002): 
Valuation of claims on nontraded assets using utility maximization, 
\emph{Mathematical Finance}, 
{\bf 12}(4), 351--373.

 \bibitem{J}
{\sc Jensen, J. L.}:  
{\it Saddlepoint Approximations}, 
{Oxford Statistical Science Series}, {\bf 16}, 1995.

 \bibitem{MY}
{\sc Mansuy, R. and M. Yor} :
{\it Aspect of Brownian Motions}.
{Springer}, 2008.

 \bibitem{MS}
{\sc Mellios, C. and P. Six} (2011): 
The traditional hedging model revisited 
with a nonobservable convenience yield, 
{\it Financial Review}, 
{\bf 46} (4), 569--593.

 \bibitem{S1}
{\sc Schwartz, E. S.} (1997): 
The stochastic behavior of commodity prices: 
implications for valuation and hedging, 
{\it The Journal of Finance}, {\bf 52}(3) , 
923--973.

 \bibitem{S2}
{\sc Schwartz, E. S.} (1998): 
Valuing long-term commodity assets,
{\it Journal of Energy Finance \& Development},  
{\bf 3}(2), 85--99. 

 \bibitem{ST}
{\sc Shiraya, K. and A. Takahashi} (2012): 
Pricing and hedging of long-term futures and forward contracts 
by a three-factor model,
{\it Quantitative Finance} {\bf 12}(12), 1811--1826.

 \bibitem{YU}
{\sc Yamauchi, H.} (2002): 
An empirical analysis of commodity future prices: corn cases, 
(in Japanese), 
{\it MTEC Journal}, {\bf 14}, 41--60.

 \bibitem{YM}
{\sc Yamamoto, H.} (2013): 
On conditionally linear asset pricing models
(in Japanese), 
{\it Master thesis, Graduate School of Engineering Science, 
Osaka University.}
\end{thebibliography}
\end{document}